\documentclass[12pt]{article}

\usepackage[margin=1in]{geometry}
\usepackage{amsthm}
\usepackage{amssymb}
\usepackage{amsmath}
\usepackage{graphicx}
\usepackage{mathdots}
\usepackage{mathtools}
\usepackage{url}
\usepackage{color}
\usepackage{framed}
\usepackage{tikz}
\usepackage{tkz-graph}

\tikzset{every node/.style={shape=circle,fill=black,inner sep=0pt, minimum size=0.5em,anchor=mid}}

\newtheorem{theorem}{Theorem}
\newtheorem{corollary}[theorem]{Corollary}
\newtheorem{lemma}[theorem]{Lemma}
\newtheorem{proposition}[theorem]{Proposition}

\theoremstyle{definition}

\title{Sketching semidefinite programs for faster clustering}

\author{Dustin~G.~Mixon\footnote{Department of Mathematics, The Ohio State University, Columbus, OH} \and Kaiying Xie\footnote{Department of Electrical and Computer Engineering, The Ohio State University, Columbus, OH}}
\date{}

\begin{document}
\maketitle

\begin{abstract}
Many clustering problems enjoy solutions by semidefinite programming.
Theoretical results in this vein frequently consider data with a planted clustering and a notion of signal strength such that the semidefinite program exactly recovers the planted clustering when the signal strength is sufficiently large.
In practice, semidefinite programs are notoriously slow, and so speedups are welcome.
In this paper, we show how to sketch a popular semidefinite relaxation of a graph clustering problem known as minimum bisection, and our analysis supports a meta-claim that the clustering task is less computationally burdensome when there is more signal.
\end{abstract}

\section{Introduction}

In many data science applications, one is tasked with partitioning objects into clusters so that members of a common cluster are more similar than members of different clusters.
For example, given a graph, one might cluster the vertices in such a way that most edges are between vertices from a common cluster.
At the same time, one ought to ensure that clusters are appropriately balanced in size, since otherwise a cluster could degenerate to a single vertex.
There are a variety of graph clustering objectives that simultaneously penalize edges across clusters and varying sizes of clusters, and for each objective, there are families of graphs for which finding an optimal clustering appears to be computationally difficult.
For example, given a graph with an even number of vertices, one might bisect the vertex set in such a way that minimizes the number of edges across clusters.
This \textit{minimum bisection problem} is known to be $\mathsf{NP}$-hard~\cite{GareyJ:79}.

Recently, it has been popular to demonstrate instances of the meta-claim that \textit{clustering is only difficult when it doesn't matter}.
For example, while one can encode hard instances of the traveling salesman problem as instances of minimum bisection, these sorts of instances would never appear in the context of real-world data science.
Rather, a data scientist will cluster data that is meant to be clustered.
With this perspective in mind, researchers have studied how a variety of clustering algorithms perform for datasets with a planted clustering structure.
For example, one might attempt to solve the minimum bisection problem for random graphs drawn from the \textit{stochastic block model}, in which edges are drawn within planted communities at a higher rate than edges across planted communities.
In this setting, it was shown in~\cite{AbbeBH:15,Bandeira:18} that a slight modification of the Goemans--Williamson semidefinite program~\cite{GoemansW:95} exactly recovers the planted clustering whenever it is information theoretically feasible to do so; see Proposition~\ref{prop.balanced sbm} for details.
Similar approaches have treated other semidefinite programming--based clustering algorithms in various settings~\cite{AwasthiEtal:15,IguchMPV:17,MixonVW:17,LiLLSW:20,LingS:19}.

Many of these results take the following form:
``Given enough signal, the semidefinite program exactly recovers the planted clusters.''
In the context of graph clustering, ``signal'' refers to the extent to which there are more edges within clusters than across clusters (in an appropriate quantitative sense).
In this paper, we pose a subtler meta-claim:
\begin{center}
\textit{Clustering is easier when there is more signal.}
\end{center}
Indeed, while previous results determined how much signal is necessary and sufficient for clustering to be computationally feasible, the above meta-claim suggests that the computational burden should decline gracefully with additional signal.
This makes intuitive sense considering it's easier to find a needle in a haystack when the haystack contains more needles.
Such behavior is particularly welcome in the context of semidefinite programming, as solvers are notoriously slow for large datasets despite having polynomial runtime.
The goal of this paper is to demonstrate this meta-claim in the special case of minimum bisection by semidefinite programming under the stochastic block model.
In other words, we provide a method to systematically decrease runtime for instances with more signal.

We start by formally defining the \textbf{stochastic block model}.
Let $G\sim\mathsf{SBM}(n_1,n_2,p,q)$ denote a random graph with vertex set $V(G)=S_1\sqcup S_2$ such that $|S_1|=n_1$ and $|S_2|=n_2$ and whose edges are independent Bernoulli random variables.
For every pair of vertices, if they reside in the same community $S_i$, we draw an edge between them with probability $p$, and otherwise the edge probability is $q$.
If $p>q$, then this models how a social network exhibits more connections within a community than between communities.
In our problem, we do not have access to the partition $\{S_1,S_2\}$.
One might randomize the communities to model this lack of information, but we do not bother with this formality here.

In the special case where $n_1=n_2$, one may exactly recover $\{S_1,S_2\}$ from $G$ provided $p$ is appropriately large compared to $q$.
To do so, we follow~\cite{AbbeBH:15,Bandeira:18} by encoding $G$ with the matrix $B$ defined by $B_{ij}=1$ if $i\leftrightarrow j$, $B_{ii}=0$, and otherwise $B_{ij}=-1$, and we then solve the program
\begin{equation}
\label{eq.original program}
\text{maximize}
\quad
x^\top Bx
\quad
\text{subject to}
\quad
1^\top x=0,
\quad
x\in\{\pm1\}^n.
\end{equation}
This corresponds to finding the minimum bisection of $G$, and one may show that the optimizers of this combinatorial program take the form $\pm(1_{\hat{S}_1}-1_{\hat{S}_2})$, where $\{\hat{S}_1,\hat{S}_2\}$ is the maximum likelihood estimator of $\{S_1,S_2\}$.
In pursuit of a computationally efficient alternative, one is inclined to consider a semidefinite program obtained by lifting $X:=xx^\top$ and relaxing:
\begin{equation}
\label{eq.afonso's sdp}
\text{maximize}
\quad
\operatorname{tr}(BX)
\quad
\text{subject to}
\quad
\operatorname{diag}(X)=1,
\quad
X\succeq 0.
\end{equation}
Impressively, this relaxation exactly recovers $\{S_1,S_2\}$ in the regime in which it is information theoretically feasible to do so:

\begin{proposition}[\cite{AbbeBH:15,Bandeira:18}]
\label{prop.balanced sbm}
Select $\alpha>\beta>0$ and for each $n\in2\mathbb{N}$, draw $G\sim\mathsf{SBM}(n/2,n/2,p,q)$ with $p=(\alpha \log n)/n$ and $q=(\beta \log n)/n$ and with planted communities $\{S_1,S_2\}$.
\begin{itemize}
\item[(a)]
If $\sqrt{\alpha}-\sqrt{\beta}>\sqrt{2}$, then \eqref{eq.afonso's sdp} recovers $\{S_1,S_2\}$ from $G$ with probability $1-o(1)$.
\item[(b)]
If $\sqrt{\alpha}-\sqrt{\beta}<\sqrt{2}$, then it is impossible to recover $\{S_1,S_2\}$ from $G$.
\end{itemize}
\end{proposition}

This phase transition partitions the set of all $(\alpha,\beta)$ into two regimes: one in which $\{S_1,S_2\}$ can be recovered from $G$ in polynomial time by semidefinite programming, and another in which no algorithm exists (not even an inefficient one) that recovers $\{S_1,S_2\}$ from $G$.
Recalling our meta-claim, we would like to show that in the former case, a larger choice of $\alpha$ for a fixed $\beta$ makes it more computationally efficient to recover $\{S_1,S_2\}$ from $G$.
In this spirit, we first consider Figure~\ref{fig.phase_transition_times}.
This figure illustrates that the SDP solver empirically behaves according to our meta-claim, taking less time to identify the clustering when there is more signal available.
In this paper, we do not explain this specific phenomenon.
Instead, we analyze a modification to the SDP algorithm using an idea that might be transferred more easily to other settings.
In fact, Figure~\ref{fig.phase_transition_times} also illustrates that our modification provides a substantial speedup over the original SDP.

\begin{figure}
\begin{center}
\includegraphics[width=\textwidth]{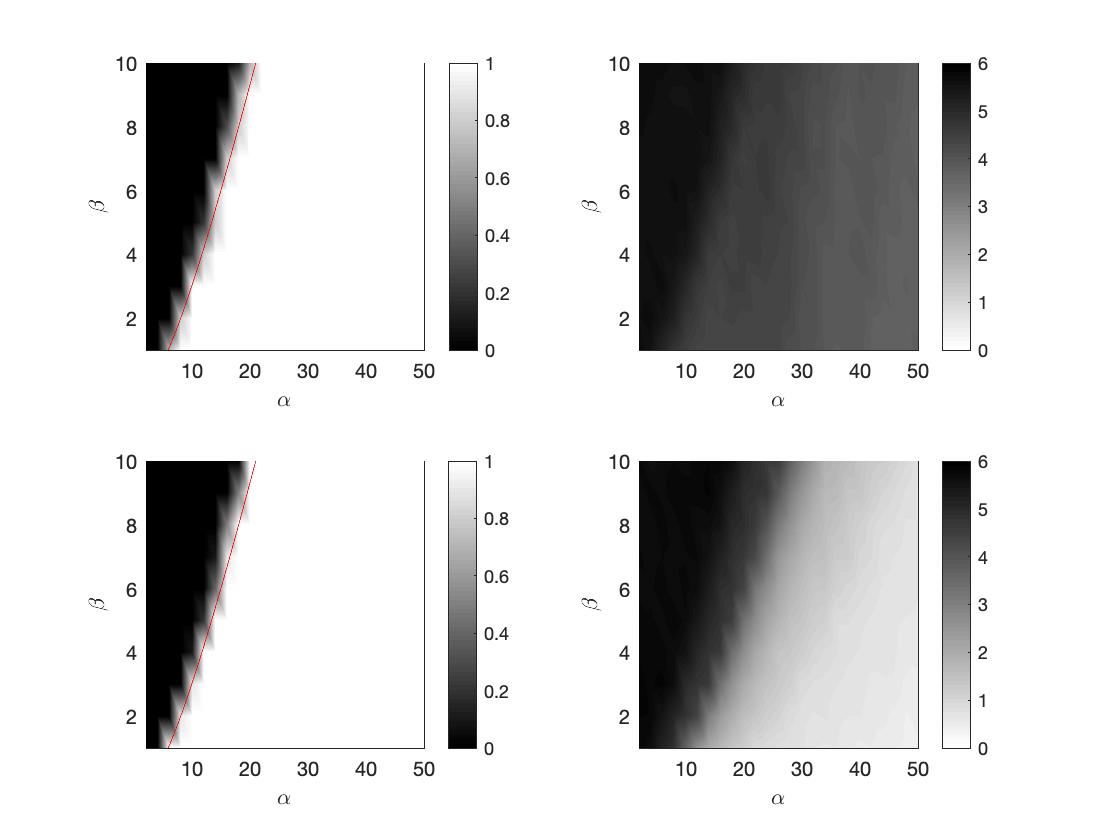}
\end{center}
\caption{\label{fig.phase_transition_times}
For each $\alpha\in\{2,4,\ldots,50\}$ and $\beta\in\{1,2,\ldots,10\}$, perform the following experiment $10$ times and plot the results.
Draw a random graph with distribution $\mathsf{SBM}(n_1=150,n_2=150,p=(\alpha\log n)/n,q=(\beta\log n)/n)$.
Attempt to exactly recover the planted communities by solving the Abbe--Bandeira--Hall SDP \eqref{eq.afonso's sdp} in CVX~\cite{GrantB:cvx}.
The proportion of recovery is displayed in \textbf{(top left)}, while the average runtime (in seconds) is displayed in \textbf{(top right)}.
Next, attempt to exactly recover the planted communities using the sketch-and-solve method described in this paper; see Section~5 for details.
The recovery rate and average runtime are displayed in \textbf{(bottom left)} and \textbf{(bottom right)}, respectively.
The red curve depicts the phase transition from Proposition~\ref{prop.balanced sbm}.
For both approaches, the runtime is smaller when $\alpha\gg\beta$, and more dramatically so for the sketch-and-solve method.
}
\end{figure}

Recall the general \textbf{sketch-and-solve approach}:
Given a large problem, we
\begin{itemize}
\item[1.]
randomly project to a smaller version of the same problem,
\item[2.]
solve the smaller version of the problem, and then
\item[3.]
use the solution to the smaller problem to (approximately) solve the original problem.
\end{itemize}
Notice that step~2 promises to be faster since the size of the problem is smaller.
This approach has been particularly effective in approximately solving large least squares problems~\cite{Woodruff:14}, and some work has been done to transfer these ideas to the setting of semidefinite programs~\cite{YurtseverUTC:17,BluhmS:19}.

We will apply the sketch-and-solve approach to systematically decrease the computational burden of clustering given more signal.
In the next section, we show how to perform steps~1 and~3 above, thereby reducing our task to solving step~2.
Specifically, we sketch by passing to the subgraph induced by a random subset of vertices, but in doing so, we no longer have communities of equal size.
This motivates the study of the unbalanced stochastic block model, and in Section~3, we show a slight modification of \eqref{eq.afonso's sdp} exactly recovers the communities in this setting.
In Section~4, we combine the ideas and results in Sections~2 and~3 to state and prove our main result.
As we will see, we can afford to sketch down to a smaller subgraph when there is more signal, which corroborates our meta-claim.
We conclude in Section~5 with a discussion.

\section{Exact recovery from a sketch oracle}

As discussed in the previous section, our approach is to sketch the original graph to a smaller graph, solve the clustering problem for the smaller graph, and then use these small clusters to determine a clustering for the entire graph.
For the last step, we will assign each vertex to the small cluster that it shares more edges with.
The following lemma indicates that this sketch-and-solve approach identifies the planted clusters provided (1) the small graph we sketch to is not too small relative to the signal (measured in terms of the sketching parameter $\gamma$ relative to the edge density parameters $\alpha$ and $\beta$) and (2) we correctly identify the planted small clusters $R_1$ and $R_2$.
We discuss how to accomplish (2) in the next section.

\begin{lemma}
\label{lem.sketching lemma}
Draw $G\sim\mathsf{SBM}(n/2,n/2,p,q)$ with planted communities $\{S_1,S_2\}$ and with $p=(\alpha\log n)/n$ and $q=(\beta\log n)/n$, where $\alpha>\beta>0$.
Draw vertices $V$ at random according to a Bernoulli process with rate $\gamma$ and put $R_i:=S_i\cap V$ for both $i\in\{1,2\}$.
Let $e(v,S)$ denote the number of edges in $G$ between $v$ and $S\subseteq V(G)$ and take
\[
\hat{S}_i
:=R_i\cup \{v\in V(G)\setminus V: e(v,R_i)>e(v,R_{3-i})\}.
\]
Then $(\hat{S}_1,\hat{S}_2)=(S_1,S_2)$ with probability $1-o(1)$ provided
\[
\gamma
>\frac{8}{3}\cdot\frac{2\alpha+\beta}{(\alpha-\beta)^2}.
\]
\end{lemma}

\begin{proof}
Let $\mathcal{E}$ denote the success event.
After conditioning on $V$, the union bound gives
\[
\mathbb{P}(\mathcal{E}^c|V)
\leq\sum_{i\in\{1,2\}}\sum_{v\in S_i}1_{\{v\in V(G)\setminus V\}}\cdot\mathbb{P}(\{e(v,R_i)\leq e(v,R_{3-i})\}|V).
\]
Conditioned on $V$, then for each $v\in S_i\cap (V(G)\setminus V)$, the quantity $e(v,R_i)-e(v,R_{3-i})$ takes the form
\[
\sum_{j=1}^{K_i}B_j^{(p)}-\sum_{j=1}^{K_{3-i}}B_j^{(q)},
\]
where $K_i:=|R_i|$ and the terms in the sums are independent Bernoulli random variables with rate indicated by the superscript.
The mean of this sum is $K_ip-K_{3-i}q$.
After centering, each term has absolute value at most $1$ almost surely, and the variance of the sum is
\[
K_ip(1-p)+K_{3-i}q(1-q)
\leq K_ip+K_{3-i}q.
\]
If in addition to $\{v\in S_i\cap (V(G)\setminus V)\}$ we restrict to the event $\mathcal{E}_i:=\{K_ip-K_{3-i}q\geq0\}$, then we may apply Bernstein's inequality for bounded variables (see Theorem 2.8.4 in~\cite{Vershynin:18}):
\begin{align*}
\mathbb{P}(\{e(v,R_i)-e(v,R_{3-i})\leq0\}|V)
&\leq \operatorname{exp}\bigg(-\frac{(K_ip-K_{3-i}q)^2/2}{K_ip+K_{3-i}q+(K_ip-K_{3-i}q)/3}\bigg)\\
&\leq \operatorname{exp}\bigg(-\frac{(K(p-q)-J(p+q))^2/2}{K(\frac{4}{3}p+\tfrac{2}{3}q)+J(\tfrac{4}{3}p-\frac{2}{3}q)}\bigg),
\end{align*}
where $K:=\frac{K_1+K_2}{2}$ and $J:=|\frac{K_1-K_2}{2}|$.
Denote $\mathcal{E}_3:=\{J\leq \epsilon K\}$ for some $\epsilon\in(0,\frac{\alpha-\beta}{\alpha+\beta})$ to be selected later.
Then on the event $\{v\in S_i\cap (V(G)\setminus V)\}\cap\mathcal{E}_i\cap\mathcal{E}_3$, it further holds that
\begin{align*}
\mathbb{P}(\{e(v,R_i)-e(v,R_{3-i})\leq0\}|V)
&\leq\operatorname{exp}\bigg(-\frac{(K(p-q)-J(p+q))^2/2}{K(\frac{4}{3}p+\tfrac{2}{3}q)+J(\tfrac{4}{3}p-\frac{2}{3}q)}\bigg)\\
&\leq\operatorname{exp}\bigg(-\frac{((p-q)-\epsilon(p+q))^2/2}{(\frac{4}{3}p+\tfrac{2}{3}q)+\epsilon(\tfrac{4}{3}p-\frac{2}{3}q)}\cdot K\bigg)
=:e_{p,q,\epsilon}(K).
\end{align*}
Overall, we may bound the failure probability:
\begin{align*}
\mathbb{P}(\mathcal{E}^c)
=\mathbb{E}[\mathbb{P}(\mathcal{E}^c|V)]
&\leq\mathbb{E}\bigg[\sum_{i\in\{1,2\}}\sum_{v\in S_i}1_{\{v\in V(G)\setminus V\}}\cdot\mathbb{P}(\{e(v,R_i)\leq e(v,R_{3-i})\}|V)\bigg]\\
&\leq\mathbb{E}\bigg[\sum_{i\in\{1,2\}}\sum_{v\in S_i}1_{\{v\in V(G)\setminus V\}}\cdot\Big(1_{(\mathcal{E}_i\cap\mathcal{E}_3)^c}+1_{\mathcal{E}_i\cap\mathcal{E}_3}\cdot e_{p,q,\epsilon}(K)\Big)\bigg]\\
&\leq\mathbb{E}\bigg[\sum_{i\in\{1,2\}}\sum_{v\in S_i}\Big(1_{(\mathcal{E}_i\cap\mathcal{E}_3)^c}+e_{p,q,\epsilon}(K)\Big)\bigg]\\
&=\frac{n}{2}\sum_{i\in\{1,2\}}\Big(\mathbb{P}((\mathcal{E}_i\cap\mathcal{E}_3)^c)+\mathbb{E}[e_{p,q,\epsilon}(K)]\Big),
\end{align*}
where the last inequality uses $e_{p,q,\epsilon}(K)\geq0$.
We will find $\epsilon\in(0,\frac{\alpha-\beta}{\alpha+\beta})$ such that
\[
\mathbb{P}(\mathcal{E}_1^c)
=e^{-\Omega(n)},
\qquad
\mathbb{P}(\mathcal{E}_2^c)
=e^{-\Omega(n)},
\qquad
\mathbb{P}(\mathcal{E}_3^c)
=e^{-\Omega(n)},
\qquad
\mathbb{E}[e_{p,q,\epsilon}(K)]
=o(1/n),
\]
from which it follows that $\mathbb{P}(\mathcal{E}^c)=o(1)$, as desired.

For the first three estimates, it will be helpful to first consider independent Bernoulli variables $X_1,\ldots,X_{n/2}$ and $Y_1\ldots,Y_{n/2}$, all with rate $\gamma$.
Given $a>b>0$, we bound
\[
f(a,b)
:=\mathbb{P}\bigg\{\sum_{j=1}^{n/2}aX_j-\sum_{j=1}^{n/2}bY_j<0\bigg\}.
\]
The expectation of the sum is $\frac{n}{2}(a-b)\gamma$, and after centering, each term in the sum has absolute value at most $a$ almost surely.
The variance of the sum is $\frac{n}{2}(a^2+b^2)\gamma(1-\gamma)$.
As such, Bernstein's inequality for bounded variables gives
\[
f(a,b)
\leq\operatorname{exp}\bigg(-\frac{(\frac{n}{2}(a-b)\gamma)^2/2}{\frac{n}{2}(a^2+b^2)\gamma(1-\gamma)+a\cdot \frac{n}{2}(a-b)\gamma/3}\bigg).
\]
Notice that $\mathbb{P}(\mathcal{E}_1^c)=\mathbb{P}(\mathcal{E}_2^c)=f(p,q)=f(\frac{\alpha\log n}{n},\frac{\beta\log n}{n})$.
Simplifying then gives
\[
\mathbb{P}(\mathcal{E}_i^c)
\leq\operatorname{exp}\bigg(-\frac{(\frac{1}{2}(\alpha-\beta)\gamma)^2/2}{\frac{1}{2}(\alpha^2+\beta^2)\gamma(1-\gamma)+\alpha\cdot \frac{1}{2}(\alpha-\beta)\gamma/3}\cdot n\bigg),
\]
which is $e^{-\Omega(n)}$ since $\alpha>\beta$.
Next,
\begin{align*}
\mathbb{P}(\mathcal{E}_3^c)
&=\mathbb{P}\{J>\epsilon K\}\\
&=\mathbb{P}\{|K_1-K_2|>\epsilon(K_1+K_2)\}\\
&\leq\mathbb{P}\{K_1-K_2>\epsilon(K_1+K_2)\}+\mathbb{P}\{K_2-K_1>\epsilon(K_1+K_2)\}\\
&=\mathbb{P}\{(1+\epsilon)K_2-(1-\epsilon)K_1<0\}+\mathbb{P}\{(1+\epsilon)K_1-(1-\epsilon)K_2<0\}\\
&=2\cdot f(1+\epsilon,1-\epsilon),
\end{align*}
which is $e^{-\Omega(n)}$ since $\epsilon\in(0,1)$.

It remains to show that $\mathbb{E}[e_{p,q,\epsilon}(K)]=o(1/n)$.
Since $K=|V|/2$, we may write
\[
e_{p,q,\epsilon}(K)
=\operatorname{exp}\bigg(-\frac{((p-q)-\epsilon(p+q))^2/4}{(\frac{4}{3}p+\tfrac{2}{3}q)+\epsilon(\tfrac{4}{3}p-\frac{2}{3}q)}\cdot |V|\bigg).
\]
We may interpret $\mathbb{E}[e_{p,q,\epsilon}(K)]$ as the moment generating function of $|V|$ evaluated at a point.
Since $|V|$ has binomial distribution over $n$ trials with rate $\gamma$, this gives
\begin{align*}
\mathbb{E}[e_{p,q,\epsilon}(K)]
&=\bigg[1-\gamma+\gamma\cdot\operatorname{exp}\bigg(-\frac{((p-q)-\epsilon(p+q))^2/4}{(\frac{4}{3}p+\tfrac{2}{3}q)+\epsilon(\tfrac{4}{3}p-\frac{2}{3}q)}\bigg)\bigg]^n
=:((1-\gamma)+\gamma e^{-c(\log n)/n})^n,
\end{align*}
where $c:=\frac{((\alpha-\beta)-\epsilon(\alpha+\beta))^2/4}{(\frac{4}{3}\alpha+\frac{2}{3}\beta)+\epsilon(\frac{4}{3}\alpha-\frac{2}{3}\beta)}$.
Next, we compare to linear approximations to obtain
\begin{align*}
((1-\gamma)+\gamma e^{-c(\log n)/n})^n
&=\operatorname{exp}(n\log((1-\gamma)+\gamma e^{-c(\log n)/n}))\\
&\leq\operatorname{exp}(n((1-\gamma)+\gamma e^{-c(\log n)/n}-1))\\
&=\operatorname{exp}(-\gamma n(1-e^{-c(\log n)/n}))
=\operatorname{exp}(-\gamma n\cdot(1-o(1))\cdot\tfrac{c\log n}{n}).
\end{align*}
Finally, by our assumption on $\gamma$, it holds that $\gamma>1/c$ for every sufficiently small $\epsilon>0$.
Taking any such $\epsilon$ gives that the above quantity is $o(1/n)$, as desired.
\end{proof}

\section{The unbalanced stochastic block model}

Lemma~\ref{lem.sketching lemma} reduces our problem of exact recovery of balanced communities in the stochastic block model to a smaller problem.
Unfortunately, by sketching to a random subgraph, the planted communities are no longer balanced, and so we cannot naively apply Proposition~\ref{prop.balanced sbm}.
Still, considering the overwhelming success of \eqref{eq.afonso's sdp} in the balanced case, one is inclined to try some version of it in the more general case.

To this end, we first note that the choice of $B=2A-J+I$ in \eqref{eq.original program} as an encoding of $G$ is somewhat arbitrary; here, $A$ denotes the adjacency matrix of $G$, and $J$ denotes the all-ones matrix.
Suppose we replaced $B$ with any $\tilde{B}=aA+b J+c I$ with $a>0$. 
Then the constraints $1^\top x=0$ and $x\in\{\pm1\}^n$ together ensure that
\[
x^\top \tilde{B}x
=x^\top(aA+b J+c I)x
=a\cdot x^\top Ax+c\cdot n.
\]
Observe that the maximizer in \eqref{eq.original program} is the same for all such objectives. 
However, we can expect different choices of $\tilde{B}$ to produce different optimizers once we relax to the SDP.
Of course, there is no change to the SDP if we change $c$ since $X$ is constrained to have all-ones diagonal, and so we take $c=0$ for simplicity.
Next, we can rescale $\tilde{B}$ so that $a=1$ without changing the SDP.
Overall, we are interested in encodings of the form $\tilde{B}=A-\mu J$.
The choice $\mu=\frac{1}{2}$ corresponds to the Abbe--Bandeira--Hall encoding~\cite{AbbeBH:15}.
Alternatively, we could run the Goemans--Williamson relaxation of maximum cut on the complement of $G$, which corresponds to taking $\mu=1$~\cite{GoemansW:95}.
These two SDPs behave similarly for $G\sim\mathsf{SBM}(n_1,n_2,p,q)$ when $n_1=n_2$, but as Figure~\ref{fig.phase_transitions} illustrates, the Abbe--Bandeira--Hall SDP performs better in the unbalanced case.
Curiously, it appears that the choice $\mu=(p+q)/2$ exhibits the phase transition from Proposition~\ref{prop.balanced sbm}.

\begin{figure}
\begin{center}
\includegraphics[width=\textwidth,trim={200 0 200 0},clip]{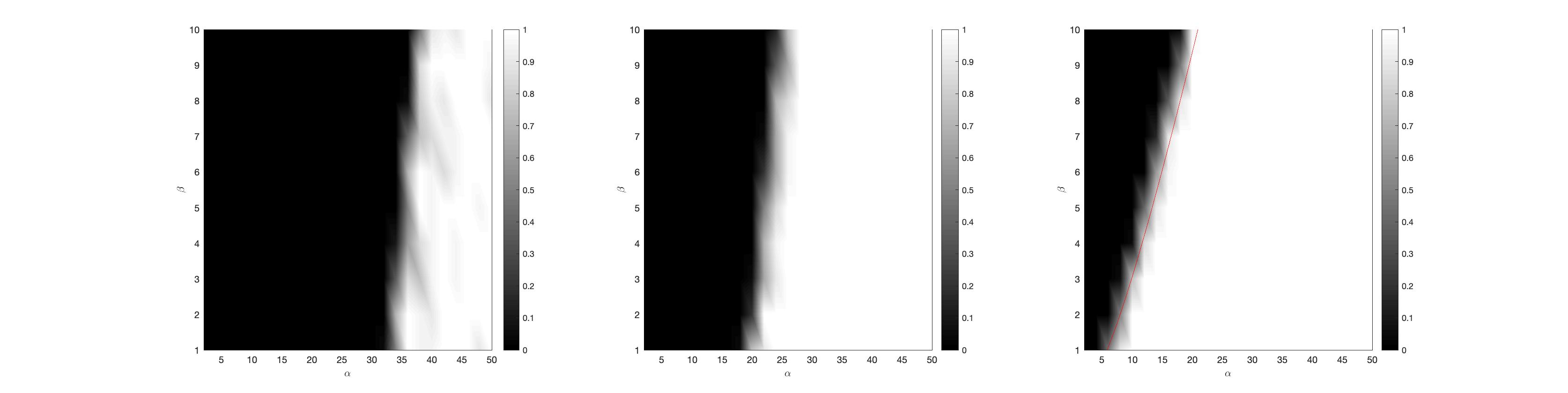}
\end{center}
\caption{\label{fig.phase_transitions}
Succes rates for exact recovery under $\mathsf{SBM}(n_1=100,n_2=200,p=(\alpha\log n)/n,q=(\beta\log n)/n)$.
\textbf{(left)}
Goemans--Williamson SDP.
\textbf{(center)}
Abbe--Bandeira--Hall SDP.
\textbf{(right)}
Proposed SDP, which assumes access to the value of $(p+q)/2$.
The red curve depicts the curve $\sqrt{\alpha}-\sqrt{\beta}=\sqrt{2}$, which by Proposition~\ref{prop.balanced sbm} is the phase transition for exact recovery in the balanced case.
It appears that the same phase transition holds for the proposed SDP in the unbalanced case.
}
\end{figure}

In this paper, we consider the family semidefinite programs
\begin{equation}
\label{eq.new sdp}
\tag{$(A,\mu)$-SDP}
\text{maximize}
\quad
\operatorname{tr}((A-\mu J)X)
\quad
\text{subject to}
\quad
\operatorname{diag}(X)=1,
\quad
X\succeq 0.
\end{equation}
Similar to~\cite{AbbeBH:15}, we pass to the dual program to obtain an optimality condition:

\begin{lemma}
\label{lem.dual cert}
Let $A$ denote the adjacency matrix of a simple graph $G$ on $n$ vertices.
Partition partition the vertices $S_1\sqcup S_2=V(G)$ and put $n_1:=|S_1|$, $n_2:=|S_2|$, and $g:=1_{S_1}-1_{S_2}$.
Let $G^+$ denote the subgraph of $G$ with edge set $E(S_1,S_1)\cup E(S_2,S_2)$, and let $G^-$ denote the subgraph with edge set $E(S_1,S_2)$.
Finally, let $D^+$ and $D^-$ denote the diagonal matrices of vertex degrees in $G^+$ and $G^-$.
If the matrix
\[
D^+-D^--\mu(n_1-n_2)\operatorname{diag}(g)-A+\mu J
\]
is positive semidefinite with rank $n-1$, then $gg^\top$ is the unique solution to $(A,\mu)$-SDP.
\end{lemma}

\begin{proof}
The dual of $(A,\mu)$-SDP is given by
\[
\text{minimize}
\quad
\operatorname{tr}(Y)
\quad
\text{subject to}
\quad
Y\succeq A-\mu J,
\quad
Y\text{ diagonal}.
\]
To verify weak duality, we have
\[
\operatorname{tr}((A-\mu J)X)
\leq\operatorname{tr}(YX)
=\operatorname{tr}(Y),
\]
where the first inequality follows from the semidefinite constraints, and the equality follows from the diagonal constraints.
By the hypotheses of the lemma, we see that
\[
X_0:=gg^\top,
\qquad
Y_0:=D^+-D^--\mu(n_1-n_2)\operatorname{diag}(g)
\]
are primal- and dual-feasible, respectively.
As such, by rearranging the above inequality, it suffices to show that equality in
\[
\operatorname{tr}((Y_0-A+\mu J)X)
\geq 0
\]
holds with primal-feasible $X$ if and only if $X=X_0$.
For the ``if'' direction, note that
\[
\operatorname{tr}(Y_0)
=\operatorname{tr}(D^+)-\operatorname{tr}(D^-)-\mu(n_1-n_2)\operatorname{tr}(\operatorname{diag}(g))
=g^\top Ag-\mu(n_1-n_2)^2,
\]
and so
\[
\operatorname{tr}((Y_0-A+\mu J)X_0)
=g^\top(Y_0-A+\mu J)g
=\operatorname{tr}(Y_0)-g^\top Ag+\mu(g^\top 1)^2
=0.
\]
For the ``only if'' direction, put $Z:=Y_0-A+\mu J$.
Then by assumption, $Z$ is positive semidefinite, and we just showed that $g^\top Zg=0$.
It follows that $Zg=0$.
Furthermore, $Z$ has rank $n-1$ by assumption, and so the nullspace of $Z$ equals the span of $g$.
Suppose $X$ is primal-feasible with $\operatorname{tr}(ZX)=0$, and consider the spectral decompositions $X=\sum_i \lambda_ix_ix_i^\top$ and $Z=\sum_j\mu_jz_jz_j^\top$.
Note that $0=\mu_1<\mu_2\leq\mu_n$ and $z_1=n^{-1/2}\cdot g$.
Then
\[
0
=\operatorname{tr}(ZX)
=\sum_i \lambda_i\sum_{j>1}\mu_j|\langle x_i,z_j\rangle|^2.
\]
Since each term in the outer sum is nonnegative, it must hold that for each $i$, either $\lambda_i=0$ or $|\langle x_i,z_j\rangle|^2=0$ for every $j>1$.
Considering $\{z_j\}_{j\in[n]}$ is an orthonormal basis for $\mathbb{R}^n$, we see that $|\langle x_i,z_j\rangle|^2=0$ for every $j>1$ only if $x_i$ is a scalar multiple of $z_1$.
On the other hand, $\{x_i\}_{i\in[n]}$ is also an orthonormal basis, and so there is at most one such $x_i$.
It follows that $X$ has rank at most $1$, and in particular, $X$ is a scalar multiple of $gg^\top$.
The diagonal constraint on the primal-feasible $X$ then implies that $X=gg^\top=X_0$, as desired.
\end{proof}

Next, we apply Lemma~\ref{lem.dual cert} to the stochastic block model to show when the SDP exactly recovers the planted clusters.
Notice the appearance of $\frac{p+q}{2}$ as a threshold on $\mu$:

\begin{lemma}
\label{lem.success probability}
Take $n_1,n_2\geq1$ and $p>q>0$, put $n:=n_1+n_2$ and $m:=|n_1-n_2|$, draw $G\sim\mathsf{SBM}(n_1,n_2,p,q)$ with planted communities $\{S_1,S_2\}$, let $A$ denote the adjacency matrix of $G$, and put $g=1_{S_1}-1_{S_2}$.
If $\mu>q$, then $gg^\top$ is the unique solution to $(A,\mu)$-SDP with probability at least
\[
\left\{\begin{array}{ll}
1-2n\operatorname{exp}(-\frac{3}{2}\cdot\frac{((\mu-q)n)^2}{(3p+q+2\mu)n+3(p-q)m})
&\text{if }\mu<\frac{p+q}{2}\vspace{8pt}\\
1-2n\operatorname{exp}(-\frac{3}{16}\cdot\frac{((p-q)n-(2\mu-(p+q))m)^2}{(2p+q)n+(2p-q-\mu)m})
&\text{if }\mu\geq\frac{p+q}{2}.
\end{array}
\right.
\]
\end{lemma}

\begin{proof}
By Lemma~\ref{lem.dual cert}, it suffices to show that with high probability, the random matrix
\[
Z
:=D^+-D^--\mu(n_1-n_2)\operatorname{diag}(g)-A+\mu J
\]
is positive semidefinite with rank $n-1$.
As established in the proof of Lemma~\ref{lem.dual cert}, it holds that $Zg=0$ almost surely.
It remains to show that the second-smallest eigenvalue $\lambda_2(Z)$ of $Z$ is strictly positive with high probability.
Weyl's inequality (see Theorem~4.3.1 in~\cite{HornJ:85}) gives
\[
\lambda_2(Z)
=\lambda_2(\mathbb{E}Z+Z-\mathbb{E}Z)
\geq \lambda_2(\mathbb{E}Z)-\|Z-\mathbb{E}Z\|_{2\to2}.
\]
To continue, we determine the exact value of $\lambda_2(\mathbb{E}Z)$, and then we use matrix Bernstein to bound $\|Z-\mathbb{E}Z\|_{2\to2}$ in a high-probability event.

To compute $\lambda_2(\mathbb{E}Z)$, it is helpful to assume (without loss of generality) that $S_1=\{1,\ldots,n_1\}$ and $S_2=\{n_1+1,\ldots,n\}$.
We first write the matrix in block form:
\[
\mathbb{E}Z
=\left[\begin{array}{cc}
(a_1-b)I_{n_1}+b1_{n_1}1_{n_1}^\top&c 1_{n_1}1_{n_2}^\top\\
c 1_{n_2}1_{n_1}^\top & (a_2-b)I_{n_2}+b1_{n_2}1_{n_2^\top}
\end{array}\right],
\]
where the matrix entries are given by
\begin{alignat*}{3}
a_1
&=p(n_1-1)-qn_2-\mu(n_1-n_2)+\mu,
&
\qquad
\qquad
b
&=-p+\mu,\\
a_2
&=p(n_2-1)-qn_1+\mu(n_1-n_2)+\mu,
&
c
&=-q+\mu.
\end{alignat*}
From this block form, we see that for each $i\in\{1,2\}$, every vector supported on $S_i$ that is orthogonal to $1_{S_i}$ is an eigenvector with eigenvalue $a_i-b$.
One may simplify to obtain
\[
a_1-b
=(p-\mu)n_1+(\mu-q)n_2,
\qquad
a_2-b
=(p-\mu)n_2+(\mu-q)n_1.
\]
Considering $g=1_{S_1}-1_{S_2}$ is an eigenvector with eigenvalue $0$, the remaining eigenvector $h$ must reside in the span of $1_{S_1}$ and $1_{S_2}$ and simultaneously be orthogonal to $1_{S_1}-1_{S_2}$.
As such, we may take $h=n_21_{S_1}+n_11_{S_2}$, and multiplying by $\mathbb{E}Z$ reveals that the corresponding eigenvalue is $(\mu-q)n$.
Observe that this is the second-smallest eigenvalue provided it is nonnegative and $p-\mu>\mu-q$, i.e., $q\leq \mu<(p+q)/2$.
On the other hand, if $\mu\geq(p+q)/2$, then the second smallest eigenvalue is the smaller of $a_1-b$ and $a_2-b$.
Overall, we have $\lambda_2(\mathbb{E}Z)\leq0$ if $\mu< q$ and otherwise
\[
\lambda_2(\mathbb{E}Z)
=\left\{\begin{array}{cl}
(\mu-q)n&\text{if }q\leq \mu<(p+q)/2,\\
(p-\mu)\max(n_1,n_2)+(\mu-q)\min(n_1,n_2)&\text{if }\mu\geq(p+q)/2.
\end{array}\right.
\]
Notice that in the case $\mu\geq(p+q)/2$, we we may simplify this expression:
\[
\lambda_2(\mathbb{E}Z)
=(p-\mu)\tfrac{n+m}{2}+(\mu-q)\tfrac{n-m}{2}
=\tfrac{p-q}{2}\cdot n-(\mu-\tfrac{p+q}{2})\cdot m.
\]

It remains to obtain a high-probability upper bound on $\|Z-\mathbb{E}Z\|_{2\to2}$.
To accomplish this, we will apply matrix Bernstein.
First, we express $Z-\mathbb{E}Z$ as a random series:
\begin{align*}
Z-\mathbb{E}Z
&=\sum_{\substack{i,j\in[n]\\i<j}}M_{ij}
:=\sum_{\substack{(i,j)\in S_1^2\cup S_2^2\\i<j}}\!\!\!\!\!\!\!(B^+_{ij}-p)(e_i-e_j)(e_i-e_j)^\top
-\sum_{\substack{i\in S_1\\j\in S_2}}(B^-_{ij}-q)(e_i+e_j)(e_i+e_j)^\top.
\end{align*}
Then $\mathbb{E}M_{ij}=0$ and $\|M_{ij}\|_{2\to2}\leq2$ almost surely for every $i,j\in[n]$ with $i<j$.
Next,
\[
\mathbb{E}M_{ij}^2
=\left\{\begin{array}{cl}
2p(1-p)(e_i-e_j)(e_i-e_j)^\top&\text{if }(i,j)\in S_1^2\cup S_2^2,~i<j\\
2q(1-q)(e_i+e_j)(e_i+e_j)^\top&\text{if }(i,j)\in S_1\times S_2,
\end{array}\right.
\]
and so we may write the sum in block form:
\[
\sum_{\substack{i,j\in[n]\\i<j}}\mathbb{E}M_{ij}^2
=\left[\begin{array}{cc}
c_1 I_{n_1} -2p(1-p)1_{n_1}1_{n_1}^\top&2q(1-q) 1_{n_1}1_{n_2}^\top\\
2q(1-q) 1_{n_2}1_{n_1}^\top&c_2 I_{n_2}-2p(1-p)1_{n_2}1_{n_2}^\top
\end{array}\right],
\]
where $c_1=2p(1-p)n_1+2q(1-q)n_2$ and $c_2=2p(1-p)n_2+2q(1-q)n_1$.
As before, every vector supported on $S_i$ that is orthogonal to $1_{S_i}$ is an eigenvector with eigenvalue $c_i$.
The other two eigenvectors reside in the span of $1_{S_1}$ and $1_{S_2}$.
Since $g$ is in the nullspace, the remaining eigenvector is $h=n_21_{S_1}+n_11_{S_2}$.
All together, the eigenvalues are
\[
\left\{\begin{array}{cl}
2p(1-p)n_1+2q(1-q)n_2&\text{with multiplicity }n_1-1\\
2p(1-p)n_2+2q(1-q)n_1&\text{with multiplicity }n_2-1\\
0&\text{with multiplicity }1\\
2q(1-q)n&\text{with multiplicity }1.
\end{array}\right.
\]
As expected, all of these eigenvalues are nonnegative.
Furthermore, since $p>q$, the largest eigenvalue is
\begin{align*}
\bigg\|\sum_{\substack{i,j\in[n]\\i<j}}\mathbb{E}M_{ij}^2\bigg\|_{2\to2}
&=2p(1-p)\max(n_1,n_2)+2q(1-q)\min(n_1,n_2)\\
&\leq2p\max(n_1,n_2)+2q\min(n_1,n_2)
=(p+q)n+(p-q)m.
\end{align*}
We are now ready to apply the matrix Bernstein inequality (see Theorem~1.6.2 in~\cite{Tropp:15}):
\[
\mathbb{P}\{\|Z-\mathbb{E}Z\|_{2\to2}\geq t\}
\leq 2n\operatorname{exp}(-\tfrac{t^2/2}{(p+q)n+(p-q)m+2t/3}).
\]
The result then follows by taking $t:=\lambda_2(\mathbb{E}Z)$ and simplifying.
\end{proof}

As an aside, we point out the following corollary:
If we know $\frac{p+q}{2}$, then we can solve the unbalanced case by semidefinite programming, though more signal is required when the communities are less balanced (at least for this result).

\begin{corollary}
Select $\alpha>\beta>0$ and $\delta\geq0$.
For each $n_1,n_2\in\mathbb{N}$ with $|n_1-n_2|\leq \delta(n_1+n_2)$, put $n=n_1+n_2$, $p=(\alpha \log n)/n$ and $q=(\beta \log n)/n$, draw $G\sim\mathsf{SBM}(n_1,n_2,p,q)$ with planted communities $\{S_1,S_2\}$, let $A$ denote the adjacency matrix of $G$, and put $g=1_{S_1}-1_{S_2}$.
Then $gg^\top$ is the unique solution to $(A,\frac{p+q}{2})$-SDP with probability $1-o(1)$ provided
\[
3(\alpha-\beta)^2
>16(2\alpha+\beta)+24(\alpha-\beta)\delta.
\]
\end{corollary}

\begin{proof}
Put $m=|n_1-n_2|$.
By Lemma~\ref{lem.success probability}, the success probability is at least
\begin{align*}
1-2n\operatorname{exp}(-\tfrac{3}{16}\cdot\tfrac{((p-q)n-(2\mu-(p+q))m)^2}{(2p+q)n+(2p-q-\mu)m})
&=1-2n\operatorname{exp}(-\tfrac{3}{16}\cdot\tfrac{((p-q)n)^2}{(2p+q)n+(3/2)(p-q)m})\\
&\geq1-2n\operatorname{exp}(-\tfrac{3}{16}\cdot\tfrac{((p-q)n)^2}{(2p+q)n+(3/2)(p-q)\delta n})\\
&=1-2\operatorname{exp}(-(\tfrac{3}{16}\cdot\tfrac{(\alpha-\beta)^2}{(2\alpha+\beta)+(3/2)(\alpha-\beta)\delta}-1)\cdot \log n),
\end{align*}
which is $1-o(1)$ by our assumption on $(\alpha,\beta,\delta)$.
\end{proof}

In the special case where $\delta=0$, this threshold matches the guarantee provided in~\cite{AbbeBH:15}. 
Perhaps surprisingly, the true phase transition appears to be independent of $\delta$; see Figure~\ref{fig.phase_transitions}.

\section{Main result}

We are now ready to state our main result, which combines our sketching approach with our solver for the sketched problem:

\begin{theorem}[main result]
\label{thm.main result}
Draw $G\sim\mathsf{SBM}(n/2,n/2,p,q)$ with planted communities $\{S_1,S_2\}$ and with $p=(\alpha\log n)/n$ and $q=(\beta\log n)/n$, where $\alpha>\beta>0$.
Consider the random variable $\mu:=|E|/\binom{n}{2}$, where $E$ denotes the random edge set of $G$.
Next, draw vertices $V\subseteq V(G)$ at random according to a Bernoulli process with rate $\gamma$, and let $A$ denote the adjacency matrix of the random subgraph induced by $V$.
In the event that the solution to $(A,\mu)$-SDP is unique and takes the form $gg^\top$ for some $g\in\{\pm1\}^V$, select $\hat{R}_1\sqcup \hat{R}_2=V$ such that $g=1_{\hat{R}_1}-1_{\hat{R}_2}$, and otherwise let $\{\hat{R}_1,\hat{R}_2\}$ be a random partition of $V$.
Next, let $\hat{S}_i$ denote $\hat{R}_i$ union the vertices in $V(G)\setminus V$ that share more edges with $\hat{R}_i$ than with $\hat{R}_{3-i}$.
Provided
\[
\gamma
>\frac{16}{3}\cdot\frac{2\alpha+\beta}{(\alpha-\beta)^2},
\]
it holds that $\{\hat{S}_1,\hat{S}_2\}=\{S_1,S_2\}$ with probability $1-o(1)$.
\end{theorem}

In words, if $\alpha\gg\beta$, then for small choices of $\gamma$, we can sketch down to a subgraph with approximately $\gamma n$ vertices before solving the SDP.
Since the runtime of the SDP is sensitive to the size of the problem instance, this promises to provide a substantial speedup.
To prove Theorem~\ref{thm.main result}, we first need a version of Lemma~\ref{lem.success probability} that accounts for the randomness in our sketching process:

\begin{lemma}
\label{lem.subset sbm recovery given mu}
Draw $G\sim\mathsf{SBM}(n/2,n/2,p,q)$ with planted communities $\{S_1,S_2\}$ and with $p=(\alpha\log n)/n$ and $q=(\beta\log n)/n$, where $\alpha>\beta>0$.
Draw vertices $V\subseteq V(G)$ at random according to a Bernoulli process with rate $\gamma$.
Let $A$ denote the adjacency matrix of the random subgraph induced by $V$.
Suppose there exists $\epsilon$ such that $\beta < \epsilon < (n/\log n)\mu  < 2\alpha-\beta$, and put $\eta:=(\tfrac{\alpha+\beta}{2}-\epsilon)_+$.
Then with probability $1-o(1)$, the solution to $(A,\mu)$-SDP is unique and identifies $\{S_1\cap V,S_2\cap V\}$ provided
\begin{equation}
\label{eq.bound on gamma}
\gamma
>\frac{16}{3}\cdot\frac{2\alpha+\beta-\eta}{(\alpha-\beta-2\eta)^2}.
\end{equation}
\end{lemma}

Before proving this lemma, we provide the idea of the proof.
There are two sources of randomness, namely, the problem instance $G\sim\mathsf{SBM}(n/2,n/2,p,q)$ and the sketching variables $B_1,\ldots,B_n\sim\mathsf{Bernoulli}(\gamma)$, all of which are independent.
Consider the subgraph $G[V]$ of $G$ induced by $V:=\{i\in[n]:B_i=1\}$.
Conditioned on the event $\{V=V_0\}$, we have
\[
G[V_0]
\sim\mathsf{SBM}(|S_1\cap V_0|,|S_2\cap V_0|,p,q)
\]
with planted communities $\{S_1\cap V_0,S_2\cap V_0\}$.
By virtue of this conditioning, we may appeal to Lemma~\ref{lem.success probability}.
Since our bound on the success probability only depends on the sizes of $S_1\cap V$ and $S_2\cap V$, we can simply condition on these sizes.
The sizes are similar with high probability, in which case the bound in Lemma~\ref{lem.success probability} gives that we succeed with high probability.
Since Lemma~\ref{lem.success probability} is broken up into cases, it is a somewhat technical exercise to make this last statement rigorous:

\begin{proof}[Proof of Lemma~\ref{lem.subset sbm recovery given mu}]
Put $g:=1_{S_1\cap V}-1_{S_2\cap V}$, and let $\mathcal{S}$ denote the event that $gg^\top$ is the unique solution to $(A,\mu)$-SDP.
We condition on $K_1:=|S_1\cap V|$ and $K_2:=|S_2\cap V|$, which are independent of each other:
\begin{equation}
\label{eq.sum to split}
\mathbb{P}(\mathcal{S}^c)
=\sum_{k_1=0}^{n/2}\sum_{k_2=0}^{n/2}\mathbb{P}(\mathcal{S}^c|\{K_1=k_1\}\cap\{K_2=k_2\})\cdot\mathbb{P}\{K_1=k_1\}\cdot\mathbb{P}\{K_2=k_2\}.
\end{equation}
Notice that $K_1,K_2\sim\mathsf{Binomial}(\tfrac{n}{2},\gamma)$, and so we expect these random variables to concentrate at $\tfrac{\gamma n}{2}$.
With this intuition, we take $\delta>0$ (to be selected later) and denote
\[
I:=\{k\in\mathbb{N}:\tfrac{(\gamma-\delta)n}{2}\leq k\leq\tfrac{(\gamma+\delta)n}{2}\}.
\]
We use this interval to split the sum \eqref{eq.sum to split} into two parts:
\begin{align}
\mathbb{P}(\mathcal{S}^c)
\nonumber
&=\sum_{(k_1,k_2)\in I\times I}+\sum_{\substack{k_1,k_2\in\{0,\ldots,n/2\}\\(k_1,k_2)\not\in I\times I}}\\
\nonumber
&\leq \sum_{(k_1,k_2)\in I\times I}\mathbb{P}(\mathcal{S}^c|\{K_1=k_1\}\cap\{K_2=k_2\})\cdot\mathbb{P}\{K_1=k_1\}\cdot\mathbb{P}\{K_2=k_2\}\\
\nonumber
&\qquad+2\cdot\mathbb{P}\{|K-\tfrac{\gamma n}{2}|>\tfrac{\delta n}{2}\}\\
\nonumber
&\leq \Big(\mathbb{P}\{K\in I\}\Big)^2\cdot\max_{(k_1,k_2)\in I\times I}\mathbb{P}(\mathcal{S}^c|\{K_1=k_1\}\cap\{K_2=k_2\})+2\cdot\mathbb{P}\{|K-\tfrac{\gamma n}{2}|>\tfrac{\delta n}{2}\}\\
\label{eq.two terms to bound}
&\leq \max_{(k_1,k_2)\in I\times I}\mathbb{P}(\mathcal{S}^c|\{K_1=k_1\}\cap\{K_2=k_2\})+2\cdot\mathbb{P}\{|K-\tfrac{\gamma n}{2}|>\tfrac{\delta n}{2}\}
\end{align}
where $K\sim\mathsf{Binomial}(\tfrac{n}{2},\gamma)$.
We bound the first term above by applying Lemma~\ref{lem.success probability}, and we bound the second term using Bernstein's inequality.

For the first term, first assume that $\mu\geq\frac{p+q}{2}$.
Then Lemma~\ref{lem.success probability} gives
\[
\mathbb{P}(\mathcal{S}^c|\{K_1=k_1\}\cap\{K_2=k_2\})
\leq 2(k_1+k_2)\operatorname{exp}(-\tfrac{3}{16}\cdot\tfrac{((p-q)(k_1+k_2)-(2\mu-(p+q))|k_1-k_2|)^2}{(2p+q)(k_1+k_2)+(2p-q-\mu)|k_1-k_2|}).
\]
We will select $\delta>0$ small enough so that
\[
(\alpha-\beta)(\gamma-\delta)
\geq (2(2\alpha-\beta)-(\alpha+\beta))\delta,
\]
which in turn implies that every $k_1,k_2\in I$ satisfies
\begin{align*}
(p-q)(k_1+k_2)
&\geq(\alpha-\beta)\tfrac{\log n}{n}\cdot (\gamma-\delta)n\\
&\geq (2(2\alpha-\beta)-(\alpha+\beta))\delta \cdot \log n
\geq (2\mu-(p+q))|k_1-k_2|.
\end{align*}
With this information, we may bound the exponent:
\[
\tfrac{((p-q)(k_1+k_2)-(2\mu-(p+q))|k_1-k_2|)^2}{(2p+q)(k_1+k_2)+(2p-q-\mu)|k_1-k_2|}
\geq\tfrac{((\alpha-\beta)(\gamma-\delta)
- (2(2\alpha-\beta)-(\alpha+\beta))\delta)^2}{(2\alpha+\beta)(\gamma+\delta)+(2\alpha-\beta-\epsilon)\delta}\cdot\log n.
\]
All together, for every $k_1,k_2\in I$, we have
\[
\mathbb{P}(\mathcal{S}^c|\{K_1=k_1\}\cap\{K_2=k_2\})
\leq 2(\gamma+\delta)\operatorname{exp}((1-\tfrac{3}{16}\cdot\tfrac{((\alpha-\beta)(\gamma-\delta)
- (2(2\alpha-\beta)-(\alpha+\beta))\delta)^2}{(2\alpha+\beta)(\gamma+\delta)+(2\alpha-\beta-\epsilon)\delta})\log n).
\]
Thanks to our assumption
\[
\gamma
>\tfrac{16(2\alpha+\beta-\eta)}{3(\alpha-\beta-2\eta)^2}
=\tfrac{16(2\alpha+\beta)}{3(\alpha-\beta)^2},
\]
this bound is $o(1)$ for every sufficiently small $\delta>0$.
Next, we consider the case in which $\mu<\frac{p+q}{2}$.
Put $\zeta:=\mu n/\log n$.
Then Lemma~\ref{lem.success probability} gives
\begin{align*}
\mathbb{P}(\mathcal{S}^c|\{K_1=k_1\}\cap\{K_2=k_2\})
&\leq 2(k_1+k_2)\operatorname{exp}(-\tfrac{3}{2}\cdot\tfrac{((\mu-q)(k_1+k_2))^2}{(3p+q+2\mu)(k_1+k_2)+3(p-q)|k_1-k_2|})\\
&\leq 2(\gamma+\delta)\operatorname{exp}((1-\tfrac{3}{2}\cdot\tfrac{((\zeta-\beta)(\gamma-\delta))^2}{(3\alpha+\beta+2\zeta)(\gamma+\delta)+3(\alpha-\beta)\delta}) \log n)\\
&\leq 2(\gamma+\delta)\operatorname{exp}((1-\tfrac{3}{2}\cdot\tfrac{((\epsilon-\beta)(\gamma-\delta))^2}{(3\alpha+\beta+2\epsilon)(\gamma+\delta)+3(\alpha-\beta)\delta}) \log n),
\end{align*}
where the last step applies $\zeta>\epsilon$ and the fact that the map $\zeta\mapsto\tfrac{((\zeta-\beta)(\gamma-\delta))^2}{(3\alpha+\beta+2\zeta)(\gamma+\delta)+3(\alpha-\beta)\delta}$ is increasing over $\zeta\in[\beta,\frac{\alpha+\beta}{2}]$.
Thanks to the assumption
\[
\gamma
>\tfrac{16(2\alpha+\beta-\eta)}{3(\alpha-\beta-2\eta)^2}
=\tfrac{16(2\alpha+\beta-(\frac{\alpha+\beta}{2}-\epsilon))}{3(\alpha-\beta-2(\frac{\alpha+\beta}{2}-\epsilon))^2}
=\tfrac{2}{3}\cdot\tfrac{3\alpha+\beta+2\epsilon}{(\epsilon-\beta)^2},
\]
this bound is $o(1)$ for every sufficiently small $\delta>0$.

At this point, we know that the first term in \eqref{eq.two terms to bound} is $o(1)$.
For the second term, note that $K-\frac{\gamma n}{2}$ is a sum of independent, mean zero random variables $X_i$ such that $|X_i|\leq1$ almost surely and $\operatorname{Var}(X_i)=\gamma(1-\gamma)\leq \gamma$.
As such, Bernstein's inequality for bounded variables gives $\mathbb{P}\{|K-\tfrac{\gamma n}{2}|>\tfrac{\delta n}{2}\}\leq 2\operatorname{exp}(-\tfrac{(\delta n/2)^2/2}{n\gamma^2+(\delta n/2)/3})=o(1)$, as desired.
\end{proof}

Next, we provide a way of estimating $\frac{p+q}{2}$ so as to select $\mu$:

\begin{lemma}
\label{lem.pick mu}
Draw $G\sim\mathsf{SBM}(n/2,n/2,p,q)$ with planted communities $\{S_1,S_2\}$ and with $p=(\alpha\log n)/n$ and $q=(\beta\log n)/n$, and consider the random variable $\mu:=|E|/\binom{n}{2}$, where $E$ denotes the random edge set of $G$.
Then for any fixed $c>0$, it holds that
\[
\Big|\mu-\frac{p+q}{2}\Big|
\leq \frac{c\log n}{n^{3/2}}
\]
with probability $1-o(1)$.
\end{lemma}

\begin{proof}
Let $A$ denote the adjacency matrix of $G$.
It is helpful to assume (without loss of generality) that $S_1=\{1,\ldots,n_1\}$ and $S_2=\{n_1+1,\ldots,n\}$.
Then
\[
\mathbb{E}A
=\left[\begin{array}{cc}
p(J-I)& qJ\\
qJ&p(J-I)
\end{array}\right],
\]
where $J$ and $I$ denote the all-ones and identity matrices of order $n/2$.
Then
\[
\mathbb{E}\mu
=\tfrac{1}{n(n-1)}\sum_{i,j=1}^n \mathbb{E}A_{ij}
=\tfrac{n-2}{2(n-1)}\cdot p+\tfrac{n}{2(n-1)}\cdot q
=\tfrac{p+q}{2}-\tfrac{p-q}{2(n-1)}.
\]
Also, the random variable $\frac{n(n-1)}{2}\mu$ is a sum of $\tfrac{n(n-1)}{2}$ independent Bernoulli random variables, $(\tfrac{n}{2})(\tfrac{n}{2}-1)$ with mean $p$ and $(\tfrac{n}{2})^2$ with mean $q$.
After subtracting the mean, each of these random variables has absolute value at most $1$ almost surely, and the variance of the sum is
\[
\operatorname{Var}(\tfrac{n(n-1)}{2}\mu)
=(\tfrac{n}{2})(\tfrac{n}{2}-1)p(1-p)+(\tfrac{n}{2})^2q(1-q)
\leq\tfrac{n(n-1)}{2}\max(p,q).
\]
We may therefore apply Bernstein's inequality for bounded variables:
\begin{align*}
\mathbb{P}\{|\mu-\mathbb{E}\mu|>t\}
&=\mathbb{P}\{|\tfrac{n(n-1)}{2}\mu-\mathbb{E}\tfrac{n(n-1)}{2}\mu|>\tfrac{n(n-1)}{2}t\}\\
&\leq 2\operatorname{exp}\bigg(-\frac{(\tfrac{n(n-1)}{2}t)^2/2}{\tfrac{n(n-1)}{2}\max(p,q)+(\tfrac{n(n-1)}{2}t)/3}\bigg)
= 2\operatorname{exp}(-\tfrac{n(n-1)}{2}\cdot\tfrac{t^2/2}{\max(p,q)+t/3}).
\end{align*}
Taking $t=\tfrac{c\log n}{2n^{3/2}}$ then gives
\[
\mathbb{P}\{|\mu-\mathbb{E}\mu|>t\}
\leq 2\operatorname{exp}(-(1-o(1))\cdot\tfrac{c^2}{16\max(\alpha,\beta) }\cdot\log n)
=o(1).
\]
All together, with probability $1-o(1)$, it holds that
\[
|\mu-\tfrac{p+q}{2}|
\leq|\mu-\mathbb{E}\mu|+|\mathbb{E}\mu-\tfrac{p+q}{2}|
\leq\tfrac{c \log n}{2n^{3/2}}+\tfrac{|\alpha-\beta|\log n}{2n(n-1)}
\leq\tfrac{c \log n}{n^{3/2}}.
\qedhere
\]
\end{proof}

\begin{proof}[Proof of Theorem~\ref{thm.main result}]
Select $\eta>0$ such that \eqref{eq.bound on gamma} is satisfied and put $\epsilon:=\frac{\alpha+\beta}{2}-\eta$.
By Lemma~\ref{lem.pick mu}, it holds that $(n/\log n)\mu\in(\epsilon,2\alpha-\beta)$ with probability $1-o(1)$.
As such, by Lemma~\ref{lem.subset sbm recovery given mu}, we have $\{\hat{R}_1,\hat{R}_2\}=\{S_1\cap V,S_2\cap V\}$ with probability $1-o(1)$.
The result then follows from Lemma~\ref{lem.sketching lemma}.
\end{proof}

\section{Discussion}

In this paper, we applied the sketch-and-solve approach to systematically reduce the computational burden of solving minimum bisection by semidefinite programming when the problem instance exhibits more signal.
Figure~\ref{fig.phase_transition_times} illustrates that our approach provides a substantial speedup over the original semidefinite programming approach.
For this figure, we selected $\gamma$ assuming access to $\alpha$ and $\beta$ using the following rule:
\[
\gamma=\min\Big\{1,\frac{4}{(\sqrt{\alpha}-\sqrt{\beta})^2}\Big\}.
\]
Based on our experiments, we expect that our sketch-and-solve approach will succeed for large $n$ in the regime
\begin{equation}
\label{eq.guess for sketching threshold}
\gamma
>\frac{2}{(\sqrt{\alpha}-\sqrt{\beta})^2},
\end{equation}
which corresponds to Proposition~\ref{prop.balanced sbm} in the special case where $\gamma=1$.

Our investigation suggests a few opportunities for future work.
First, while Theorem~\ref{thm.main result} provides a sufficient condition on $\gamma$, the apparent sketching threshold \eqref{eq.guess for sketching threshold} warrants an explanation.
Also, we did not provide a method for selecting the sketching parameter $\gamma$ when $\alpha$ and $\beta$ are unknown.
In practice, one can test whether a given $\gamma$ is too small by running the algorithm multiple times and observing whether substantially different clusterings emerge.
Presumably, one can identify a good choice of $\gamma$ by performing this test for increasing values of $\gamma$, but it would be nice to have a principled approach to select $\gamma$.
Next, Figure~\ref{fig.phase_transitions} indicates that the phase transition from Proposition~\ref{prop.balanced sbm} emerges when solving exact recovery under the \textit{unbalanced} stochastic block model.
A more detailed analysis is required to explain this phenomenon.
To obtain this phase transition empirically, we put $\mu=\frac{p+q}{2}$, which suggests another question:
Given $G\sim\mathsf{SBM}(n_1,n_2,p,q)$ with unknown $n_1,n_2,p,q$, how does one estimate $\frac{p+q}{2}$?
In our setting, the fact that $n_1=n_2$ before sketching helped us to estimate $\frac{p+q}{2}$ in Lemma~\ref{lem.pick mu}.
Finally, it would be interesting if sketching SDPs in other settings (such as geometric clustering) also supports our meta-claim that clustering is easier in the presence of more signal.

\section*{Acknowledgments}

DGM was partially supported by AFOSR FA9550-18-1-0107 and
NSF DMS 1829955.

\end{document}